\documentclass[11pt]{article}
\usepackage{amsmath}
\usepackage{amsfonts,amsthm,amssymb}
\usepackage{amstext}
\usepackage{graphicx}
\usepackage{geometry}
\usepackage{enumitem}
\usepackage{color}
\usepackage{algorithm}
\usepackage{algpseudocode}
\usepackage{multirow}

\newcommand{\set}[1]{\left\{#1\right\}}

\newcommand{\bT}{{\mathbb{T}}}
\newcommand{\Real}{\mathbb{R}}
\newcommand{\Ze}{\mathbb{Z}}
\newcommand{\BFu}{u}
\newcommand{\BFs}{s}

\newtheorem{lemma}{Lemma}[section]
\newtheorem{theorem}[lemma]{Theorem}

\newtheorem{corollary}[lemma]{Corollary}
\newtheorem{proposition}[lemma]{Proposition}

\theoremstyle{definition}
\newtheorem{definition}[lemma]{Definition}
\newtheorem{remark}[lemma]{Remark}

\newcounter{property}
\newenvironment{properties}
{
\addtocounter{property}{1}
\begin{enumerate}[labelindent=0pt,label=(\Alph{property}\arabic*),itemindent=1em]
}
{
\end{enumerate}
}

\title{On the computational complexity of minimum-concave-cost flow in a two-dimensional grid}
\author{Shabbir Ahmed\thanks{H. Milton Stewart School of Industrial \& Systems Engineering, Georgia Institute of Technology, Atlanta GA 30332, USA. Email: sahmed@isye.gatech.edu.}, Qie He\thanks{Department of Industrial \& Systems Engineering, University of Minnesota, Minneapolis MN 55455, USA. Email: qhe@umn.edu.}, Shi Li\thanks{Department of Computer Science and Engineering, The State University of New York at Buffalo, Buffalo, NY 14260, USA. Email: shil@buffalo.edu.}, George L. Nemhauser\thanks{H. Milton Stewart School of Industrial \& Systems Engineering, Georgia Institute of Technology, Atlanta GA 30332, USA. Email: gnemhaus@isye.gatech.edu.}}
\date{}

\begin{document}
\maketitle
\begin{abstract}
We study the minimum-concave-cost flow problem on a two-dimensional grid. We characterize the computational complexity of this problem based on the number of rows and columns of the grid, the number of different capacities over all arcs, and the location of sources and sinks. The concave cost over each arc is assumed to be evaluated through an oracle machine, i.e., the oracle machine returns the cost over an arc in a single computational step, given the flow value and the arc index. We propose an algorithm whose running time is polynomial in the number of columns of the grid, for the following cases: (1) the grid has a constant number of rows, a constant number of different capacities over all arcs, and sources and sinks in at most two rows; (2) the grid has two rows and a constant number of different capacities over all arcs connecting rows; (3) the grid has a constant number of rows and all sources in one row, with infinite capacity over each arc. These are presumably the most general polynomially solvable cases, since we show the problem becomes NP-hard when any condition in these cases is removed. Our results apply to abundant variants and generalizations of the dynamic lot sizing model, and answer several questions raised in serial supply chain optimization.
\end{abstract}

\section{Introduction}
We are interested in the min-cost network flow problem over an $L$-by-$T$ grid $G$. The grid $G$ is an acyclic directed planar graph on vertices $V=\{v_{l,t} \mid l=1,\ldots, L; \; t =1, \ldots, T\}$ with arc set 
\begin{equation*}
A=\{(v_{l,t},v_{l+1,t})\mid l\le L-1\}\bigcup \{(v_{l,t},v_{l,t+1})\mid t \le T-1\}.
\end{equation*}
A drawing of $G$, embedded in $\Real^2$, is shown in Figure~\ref{fig:grid}. 
\begin{figure}[htb]
\centering
\includegraphics[scale=0.5]{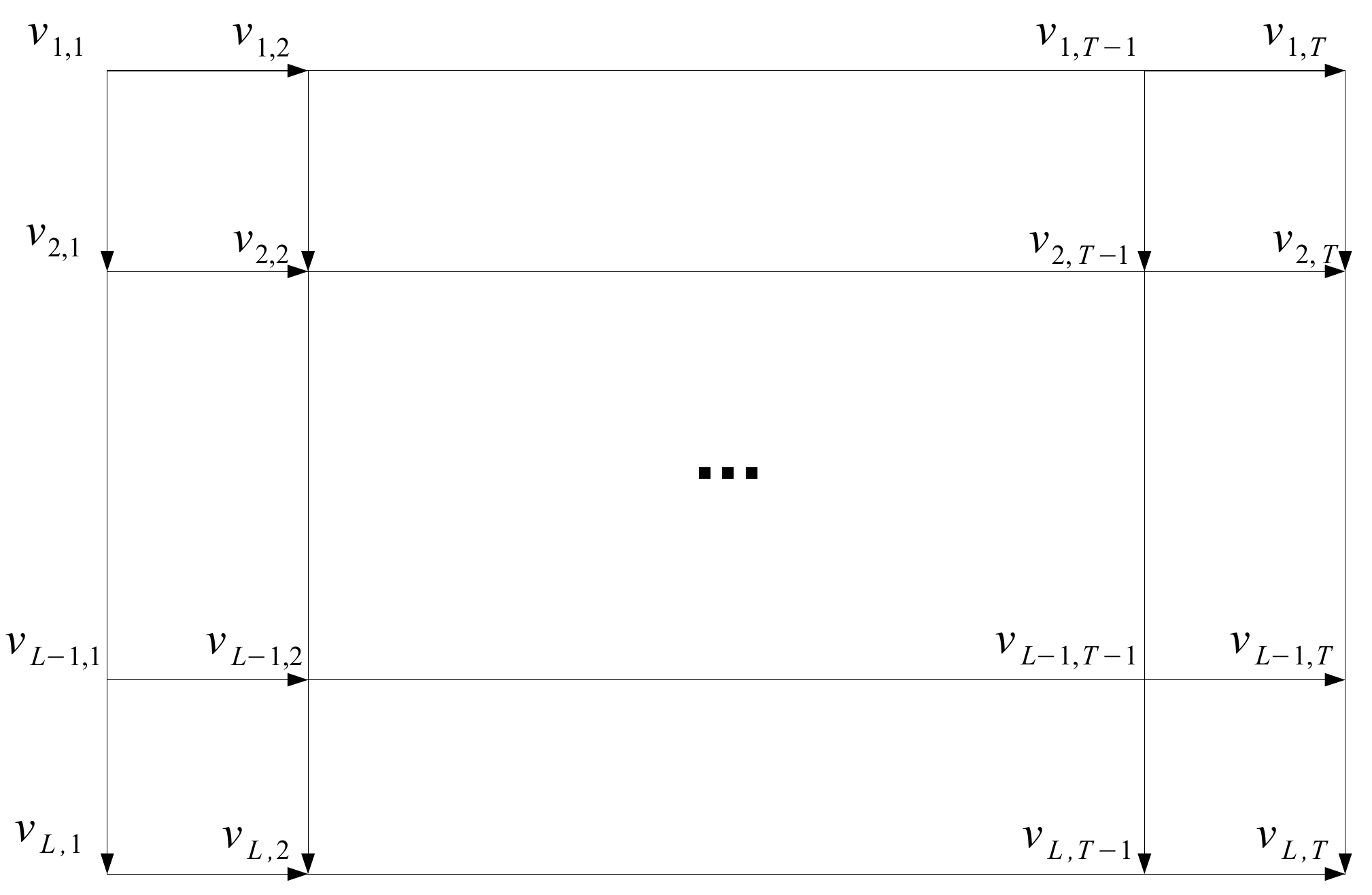}
\caption{The $L$-by-$T$ grid} 
\label{fig:grid}
\end{figure}
The grid $G$ has some additional parameters: a supply vector $b = (b_{v}) \in \Ze^{|V|}$, a cost function $c_a: \Real \rightarrow \Real$ for each arc $a\in A$, and an arc-capacity vector $U = (U_a) \in \overline{\Ze}_+^{|A|}$, where $\overline{\Ze}_+$ is the set of non-negative integers and infinity. We assume the net supply $\sum_{i\in V}b_v=0$. We call a vertex $v$ a source if $b_v>0$, a sink if $b_v<0$, or a transshipment vertex otherwise. We assume the function $c_a$ over each arc $a$ is represented by an oracle machine. Given the flow value $x$ over arc $a$, the oracle machine returns the cost $c_a(x)$ in a single computational step, so cost functions with explicit forms are special cases in our oracle model. 

The minimum-concave-cost flow problem in a two-dimensional grid (MFG) is to find a vector $x \in \Real^{|A|}$ to
\begin{equation} \label{eq:FlowConstraints}
\begin{array}{lll}
\mathrm{minimize} & \sum_{a\in A}{c_a(x_a)} &\\
\text{s.t.} &  \sum_{a\in \delta^+(v)}{x_a}-\sum_{a\in \delta^-(v)}{x_a} = b_v,& \forall v\in V, \\
						&  0\le x_a\le U_a, & \forall a \in A,
\end{array}
\end{equation}
where $c_a$ is concave for each $a\in A$, and $\delta^+(v)$ and $\delta^-(v)$ are the set of outgoing and incoming arcs of vertex $v$, respectively. A \emph{feasible flow} is a vector $x \in \Real^{|A|}$ that satisfies the constraints in~\eqref{eq:FlowConstraints}. 

Our interest in the computational complexity of the MFG is kindled by the classical dynamic uncapacitated lot sizing problem (ULSP), a building block in production planning and inventory control, studied by Wagner and Whitin~\cite{wagner1958dynamic}. The ULSP aims to find a min-cost production schedule to satisfy a given sequence of demands over $T$ periods, with no limit on production capacity. It is a special case of the MFG with two rows and one source~\cite{he2012minimum}, and can be solved by dynamic programming (DP) efficiently in $O(T\ln T)$ time~\cite{aggarwal1993improved,federgruen1991simple,wagelmans1992economic}. Many efforts have been spent on extending the ULSP to models under more practical settings, such as multiple production steps, demands at intermediate production steps, capacity on production or storage, and different production and storage costs which characterize the setup cost or economies of scale. All these variants can be transformed into the MFG with \emph{a single source} by adjusting the corresponding parameters in~\eqref{eq:FlowConstraints}, such as $L$, $U_a$, and the location of sinks~\cite{he2012minimum}. The computational complexity, however, varies across models. For example, the constant capacitated lot sizing problem, an extension of the ULSP with a uniform production capacity at all periods, can be solved in $O(T^3)$ time~\cite{van1996t3}; In contrast, the capacitated lot sizing problem, where the production capacities can be arbitrary across periods, is NP-hard. Let 
\[K=|\{U_a \mid U_a < +\infty, a\in A\}|\]
be the number of different arc capacity values of the grid. Our goal is to study how the network parameters, in particular \emph{the number of rows $L$, the number of arc capacity values $K$, and the location of sources and sinks}, affect the computational complexity of the MFG. In particular, we aim to identify the most general conditions under which the MFG can be efficiently solved.

Since the arc capacity plays such a significant role in the computational complexity of the MFG, we divide the MFG into two classes: the uncapacitated MFG (U-MFG), where each arc has infinite capacity or $K=0$, and the capacitated MFG (C-MFG), where \emph{at least} one arc has a finite capacity or $K\ge 1$. We summarize the complexity of the two classes in Tables~\ref{table:ccfg} and~\ref{table:ufg} below, with newly derived results in this paper highlighted in bold. 
\begin{table}[htb]
\caption{Complexity of the C-MFG with constant $K$}
\centering
\begin{tabular}{c||p{5cm}|p{5cm}}
\hline\hline
 & Sources and sinks in at most two rows & Sources and sinks in at least three rows\\
\hline
Constant $L$  & \textbf{Polynomially solvable}  & \textbf{NP-hard}\\
\hline
Varying $L$ & Open & \textbf{NP-hard} \\
\hline
\end{tabular}
\label{table:ccfg}
\end{table}

\begin{table}[htb]
\caption{Complexity of the U-MFG}
\centering
\begin{tabular}{ p{2cm} | p{4.5cm} | p{4.5cm}|p{3.5cm}}
\hline\hline
 \multicolumn{3}{c|}{Sources in one row} & Both sources and sinks in at least two rows\\
\hline
\multirow{3}{*}{Constant $L$} & Sinks in one row & Polynomially solvable~\cite{erickson1987send} & \multirow{4}{*}{\textbf{NP-hard}}\\
& Sinks in at most two rows & Polynomially solvable~\cite{he2012minimum} & \\
& Sinks in at most $L$ rows  &  \textbf{Polynomially solvable} &  \\
\hline
Varying $L$ & Sinks in at most $L$ rows & \textbf{NP-hard} & \textbf{NP-hard} \\
\hline
\end{tabular}
\label{table:ufg}
\end{table}

In Table~\ref{table:ccfg}, we only list the results for constant $K$, since with varying $K$ the capacitated lot sizing problem, a special case of the C-MFG, is already NP-hard.
As we see from the two tables, each condition for the polynomially solvable cases is crucial; any relaxation of these conditions renders the problem NP-hard. Our result answers several questions raised in the context of supply chain optimization. In~\cite{kaminsky2003production, van2005integrated}, polynomial-time algorithms were derived for a multi-stage serial supply chain design problem, in which capacity is only present at the manufacturing stage (the vertical arcs between the first and second rows in the MFG); the authors raised the question of whether an optimal design can be found efficiently, if capacity is imposed at other stages. Our result in Table~\ref{table:ccfg} indicates that the optimal design can always be found in polynomial time, as long as the number of different capacities is a constant, regardless of at which stages the capacity is imposed. Another open problem was raised in~\cite{he2012minimum} on whether the U-MFG with sources in one row and sinks in $L>2$ rows is polynomially solvable. In this paper, we give a positive answer to the problem when $L$ is an arbitrary positive integer, thus generalizing the result in~\cite{he2012minimum} from sinks in at most two rows to sinks in an arbitrary number of rows; we also complement the polynomially solvable case with a hardness result when $L$ is allowed to vary.

We should also mention the model of computation used in this paper. Since the cost over an arc in a feasible flow may be a real number, any algorithm for the MFG will involve computation over real numbers. We adopt the real Random Access Machine (RAM) introduced in~\cite{blum1989theory} as our model of computation; the real RAM model has registers each of which is capable of storing a real number with infinite precision, and each arithmetic operation on real numbers stored in these registers, such as adding or comparing two real numbers, takes a single computational step. The running time of our algorithm refers to the total number of arithmetic operations and oracle queries.

We now survey some results on minimum-concave-cost flow problem in the literature. The minimum-concave-cost flow problem over a general network can be shown to be NP-hard, proven by a reduction from the partition problem~\cite{guisewite1990minimum}. Approximation and exact algorithms based on DP or branch-and-bound were developed for various network topologies and cost functions; see the survey in~\cite{guisewite1990minimum} and~\cite{fontes2006dynamic}. Polynomially solvable special cases of the minimum-concave-cost flow problem include a single-source problem with a single nonlinear arc cost~\cite{guisewite1993polynomial}, the problem with a constant number of sources and nonlinear arc costs~\cite{tuy1995minimum}, a production-transportation network flow problem with concave costs defined on a constant number of variables~\cite{tuy1996strongly}, the pure remanufacturing problem~\cite{van2008four}, and many variants and extensions of the ULSP, which we elaborate in the next paragraph. We also point out here two relatively general classes of polynomially solvable minimum-concave-cost flow problem. Erickson et al.~\cite{erickson1987send} derived a DP algorithm that runs in polynomial time in a planar graph with sources and sinks lying on a constant number of faces of the graph; thus their algorithm runs in polynomial time for the U-MFG with sources in the first row and sinks in the last row, but runs in exponential time for the U-MFG with sources or sinks in more than one rows and the C-MFG in general. Recently, He et al.~\cite{he2012minimum} gave a polynomial-time DP for the U-MFG with a constant number of rows, sources in one row, and sinks in two other rows, and the U-MFG with one source and sinks in a constant number of rows.

The literature on the lot sizing problem is abundant; see the book by Pochet and Wolsey~\cite{pochet2006production}. To be concise, we focus our review on problems that can be formulated as an MFG. We also adopt the terminology in the lot sizing literature: \emph{a time period} refers to a column, and \emph{an echelon} or \emph{a stage} refers to a row in the MFG. For the uncapacitated case, the ULSP was first solved in $O(T^2)$ time by Wagner and Whitin~\cite{wagner1958dynamic}, and the complexity was later improved to $O(T\ln T)$~\cite{aggarwal1993improved,federgruen1991simple,wagelmans1992economic}; Zangwill gave an $O(LT^4)$-time DP for the multi-echelon ULSP with sinks at the last echelon~\cite{zangwill1969backlogging}; Love~\cite{love1972facilities} gave an $O(LT^3)$-time DP algorithm for the multi-echelon ULSP if the production costs are non-increasing over time periods, and the storage costs are non-decreasing over echelons; Zhang et al.~\cite{zhang2011polyhedral} recently proposed an $O(T^4)$-time DP for the two-echelon ULSP with two rows of sinks and fixed setup costs for production, and derived a new class of valid inequalities for the multi-echelon ULSP. For lot sizing problems with production or storage capacity, the capacitated lot sizing problem is NP-hard~\cite{bitran1982computational}, but the constant capacitated lot sizing problem can be solved in $O(T^3)$ time~\cite{florian1971deterministic,van1996t3}; the lot sizing problem with variable storage bounds and fixed-charge storage costs can be solved in $O(T^2)$ time~\cite{atamturk2005lot,atamturk2008n2}. 

Some extensions of the ULSP are presented in a more general context of supply chain optimization. Kaminsky and Simchi-Levi~\cite{kaminsky2003production} studied a two-stage ($L=3$ in our model) supply chain management problem with constant capacity at the first stage, and gave polynomial-time algorithms when the production and transportation costs have several different structures. Van Hoesel et al.~\cite{van2005integrated} extended this two-stage model to a multi-echelon serial supply chain model with general concave production and storage costs, with constant capacity at the first echelon and unlimited capacity at other echelons; they gave an $O(LT^{2L+3})$-time DP, and the complexity can be significantly reduced if the production and storage costs follow some special patterns. Their result was recently improved by Hwang et al.~\cite{hwang2013basis}, who developed an $O(LT^8)$-time algorithm by exploiting a structure called basis path in the optimal solutions. 

To the best of our knowledge, all the polynomially solvable capacitated lot sizing and serial supply chain models in the literature make the following two assumptions: (1) the finite capacity is restricted to the production echelon, and there is unlimited capacity at other echelons; (2) the capacity is uniform at all time periods. In our model, this means each vertical arc between the first and second rows of the grid has the same finite capacity, while other arcs have infinite capacity. These assumptions are rather restrictive from a practical point of view. Capacities arise naturally at other echelons of the supply chain, such as transportation between the manufacturer and distribution centers or between distribution centers and retailers, and capacities may vary across time periods as well. It was also pointed out in~\cite{kaminsky2003production,van2005integrated} that production planning models with more general capacity structure is a future research direction and their complexities were posed as open questions. Our result indicates that as long as the number of different capacities over all arcs is a constant, the problem can be solved in polynomial time, regardless of at which stages or upon which vertical or horizontal arcs the capacity is imposed.

The rest of the paper is organized as follows. In Section~\ref{sec:DP}, we reformulate the MFG as an optimal control problem over a dynamical system, propose a general algorithm to solve the problem, and highlight the key component that affects the complexity of the algorithm. In Section~\ref{sec:CFG}, we identify two conditions for the C-MFG to be polynomially solvable, and then present several NP-hard cases with any of the conditions relaxed. In Section~\ref{sec:UCFG}, we identify a general condition for the U-MFG to be polynomially solvable, and complement it with two NP-hard cases. Section~\ref{sec:extension} discusses extensions of the MFG to grids with additional arcs. Conclusion and some open problem are given in Section~\ref{sec:conclude}.

\section{Methodology} \label{sec:DP}
We introduce some notations and terminology first. Given an integer $N$, let $[N]$ denote the set $\{1,\ldots,N\}$. Let the set $P_F$ denote the \emph{flow polyhedron} defined by constraints in~\eqref{eq:FlowConstraints}. Let $b_{l,t}$ denote the supply of vertex $v_{l,t} \in V$. We call any arc $(v_{l,t},v_{l,t+1})$ with $l\in [L]$ and $t\in [T-1]$ a \emph{forward arc}, and any arc $(v_{l,t}, v_{l+1,t})$ with $l\in [L-1]$ and $t\in [T]$ a \emph{downward arc}. Given a tree $\bT=(V_{\bT}, A_{\bT})$ and a vertex $v$, for simplicity we sometimes write $v\in \bT$ instead of $v\in V_{\bT}$ to denote that $v$ is a vertex in $\bT$. We use $G$ to denote both the grid as well as its planar embedding in Figure~\ref{fig:grid}. The outer face of $G$ refers to the unbounded region outside the rectangle with vertices $v_{1,1}, v_{1,T}, v_{L,1}$, and $v_{L,T}$ in Figure~\ref{fig:grid}.

\subsection{Problem reformulation and the algorithm}
We always assume the given MFG instance is feasible, since its feasibility can be checked in the same way as in a min-linear-cost flow problem, by solving a maximum flow problem~\cite{ahuja1993network}. In order to describe our algorithm, we first reformulate the MFG as an optimal control problem for a discrete-time linear dynamical system with concave costs. The elements of the dynamical system are as follows.
\begin{enumerate}
	\item Decision stages. There are $T+1$ stages, corresponding to column $t=1,\ldots,T$ in the grid and a dummy stage 0 in the beginning.	 
	\item States. The state $\BFs^t$ at stage $t\in [T-1]$ is an $L$-dimensional vector whose component $s^t_l$ denotes the flow over the forward arc $(v_{l,t},v_{l,t+1})$, for $l\in [L]$. Each of the states $\BFs^0$ and $\BFs^T$ is an $L$-dimensional zero vector. In addition, the dimension of $\BFs^t$ for $t\in [T-1]$ can be reduced by one, since $\sum_{l\in [L]} s^t_l= \sum_{i\in [t]}\sum_{l\in [L]}b_{l,i}$ according to the flow balance constraints. 	
	\item Decision variables (actions, controls, or inputs). The decision variable $\BFu^t$ at stage $t\in [T]$ is an $(L-1)$-dimensional vector whose component $u^t_l$ is the flow over the downward arc $(v_{l,t+1},v_{l+1,t+1})$, for $l\in [L-1]$.	
	\item The system equations. The state $\BFs^{t+1}$ at stage $t+1$ can be easily determined by $\BFs^t$ and $\BFu^t$ through the flow balance constraints for vertices $v_{1,t+1}, \ldots, v_{L,t+1}$. For simplicity, we assume the system equation at stage $t$ is $\BFs^{t+1}=H_t(\BFs^t,\BFu^t)$, where $H_t$ is the affine function representing the flow balance constraints for vertices in column $t+1$.
	\item The cost. The cost at stage $t\in [T]$ is the sum of costs incurred by the downward and forward arcs related to that stage. In particular, the cost is $\sum_{l \in [L]} c^f_{l,t}(s^t_l) + \sum_{l\in [L-1]} c^d_{l,t}(u^t_l)$, where $c^f_{l,t}$ is the cost function over forward arc $(v_{l,t}, v_{l,t+1})$, and $c^d_{l,t}$ is the cost function over downward arc $(v_{l,t},v_{l+1,t})$.
\end{enumerate}
The above dynamical system is different from the system in~\cite{he2012minimum}, where the stages are the anti-diagonal lines of the grid and each component of the state represents the inflow of some vertex. The current dynamical system has the advantage of being easily extended to a grid with horizontal backward arcs and vertical upward arcs. Basically, we only need to augment the state vector and decision variables and change the system equations and costs accordingly. We will discuss several extensions of the MFG in detail in Section~\ref{sec:extension}

The challenge of designing the optimal control over such a dynamical system comes from the uncountable state space and the concave cost structure. The first simplification we can do is to reduce the state space of the system to a finite set. Then the optimal control problem is equivalent to a shortest path problem over an acyclic graph~\cite{bertsekas1996dynamic}. The simplification is based on the following observation: since the MFG involves minimizing a concave function over the flow polyhedron $P_F$, so the optimum must be attained at one of its extreme points~\cite{rockafellar2015convex}. Thus it suffices to consider only the states corresponding to those extreme points, and the number of extreme points for any given polyhedron is always finite. We now describe below the algorithm for the MFG.

\begin{algorithm}
\caption{The algorithm for the MFG}
\label{alg:mfg}
\begin{enumerate}
	\item Enumerate all the possible values of $\BFs^t$ corresponding to the extreme points of $P_F$, for each $t\in [T-1]$.
	\item Construct a $(T+1)$-partite directed graph $G'=(V', A')$. The vertex set $V'=\cup_{t=0}^T V'_t$, where each vertex in $V'_t$ represents a possible value of state $s^t$ (the set $V'_0$ and $V'_{T}$ contain exactly one vertex, respectively, since both $\BFs^0$ and $\BFs^T$ equal the zero vector); each arc in $A'$ goes from one vertex in $V'_{t}$ to one vertex in $V'_{t+1}$ for $t=0,\ldots,T-1$. Suppose vertex $i$ in $V'_t$ represents a state $s^{t,i}$ at stage $t$, vertex $j$ in $V'_{t+1}$ represents a state $s^{t+1,j}$ at stage $t+1$, and $u^{t,ij} \in \Real^{L-1}_+$ is the corresponding decision variable computed through the system equations $s^{t+1,j} = H_t(s^{t,i}, u^{t,ij})$, then the cost of arc $(i,j)$ in $G$ requires $2L-1$ oracle queries and is calculated as follows:
\begin{equation*}
c_{ij} =  \sum_{l\in [L]} c^f_{l,t}(s^{t,i}_l) + \sum_{l\in [L-1]} c^d_{l,t}(u^{t,ij}_l),
\end{equation*}
if each component of $u^{t,ij}$ satisfies the corresponding arc capacity constraint, and $c_{ij}=\infty$ otherwise. 
	\item Find a shortest path $P$ from the vertex representing $\BFs^0$ to the vertex representing $\BFs^T$ in $G'$. 
	\item Recover a flow of the MFG from the shortest path $P$, by solving the system equations $\BFs^{t+1} = H_t(\BFs^{t}, \BFu^t)$ with $\BFs^t$ being the value corresponding to the vertex in $V'_t$ in $P$, for $t=0,\ldots,T-1$.
\end{enumerate}
\end{algorithm}

\begin{theorem} \label{prop:runtime}
Algorithm~\ref{alg:mfg} solves the MFG in $O(LTM^{2L-2})$ time, where $M$ is the maximum number of flow values a forward arc can attain in all extreme points of $P_F$.
\end{theorem}
\begin{proof}
Any extreme point of $P_F$ corresponds to a feasible path from the vertex representing $\BFs^0$ to the vertex representing $\BFs^T$ in $G'$, whose length equals the cost of that extreme point; conversely, the shortest path $P$ found by Algorithm~\ref{alg:mfg} corresponds to a feasible flow in the MFG whose cost is equal to the length of the shortest path. Thus the flow recovered from $P$ is an optimal flow for the MFG.

Next we analyze the running time of Algorithm~\ref{alg:mfg}. Suppose $M$ is the maximum number of flow values a forward arc can attain in all extreme points of $P_F$. Since state $\BFs^t$ is an $(L-1)$-dimensional vector with each component being the flow over some forward arc, the number of states at stage $t$ in all extreme points of $P_F$ is $O(M^{L-1})$, and these states can be enumerated in $O(M^{L-1})$ time as well. Then the graph $G'$ has $O(TM^{L-1})$ vertices, and the number of arcs is
\[O(M^{L-1}) + \underbrace{O(M^{2L-2}) + \ldots + O(M^{2L-2})}_{T-2 \text{ terms}} + O(M^{L-1})=O(TM^{2L-2}),\]
where each term in the summation above is the number of arcs between $V'_{t-1}$ and $V'_{t}$ for $t \in [T]$. The cost of each arc can be evaluated in $O(L)$ time with $O(L)$ oracle queries, so the construction of graph $G'$ takes $O(LTM^{2L-2})$ time including $O(LTM^{2L-2})$ oracle queries. Since graph $G'$ is a directed acyclic graph, the shortest path can be found efficiently in $\Theta(|V'|+|A'|) = O(TM^{2L-2})$ time~\cite{cormen2009}. It takes $O(LT)$ time to recover of the flow from the found shortest path. The overall running of Algorithm~\ref{alg:mfg} is $O(LTM^{2L-2})$.
\end{proof}

\noindent We should mention that although each component of the state corresponds to a flow value in some extreme point of $P_F$, the optimal flow output by Algorithm~\ref{alg:mfg} is not necessarily an extreme point, since it is possible for the free arcs in an optimal flow to contain an undirected cycle, unless the optimum is unique. 

The remaining task, the key challenge for the MFG, is to find out all possible flow values over a forward arc in all extreme points of $P_F$, and derive a bound on the value $M$.

\subsection{Characterization of the extreme point of $P_F$}
We first introduce some terminology related to flow and extreme points of $P_F$, adopted from~\cite{ahuja1993network}. Given a feasible flow $f$, we call an arc $a$ a \emph{free arc} if $0< f_a < U_a$ and a \emph{restricted arc} if $f_a=0$ or $f_a=U_a$. 

\begin{definition}~\cite{ahuja1993network}
A feasible flow $f$ in graph $G$ is a \emph{cycle free solution}, if $G$ contains no undirected cycle composed only of free arcs. 
\end{definition} 

\begin{definition}~\cite{ahuja1993network}
A feasible flow $f$ and an associated spanning tree in $G$ is a \emph{spanning tree solution}, if every nontree arc in $G$ is a restricted arc for $f$. 
\end{definition}
\noindent Note that in a spanning tree solution, the tree arcs can be free or restricted. A tight connection exists between the extreme point of $P_F$ and a cycle free solution (spanning tree solution).
 
\begin{proposition} [Chapter 7.3 in~\cite{bertsimas1997introduction}] \label{prop:cyclefree}
A feasible flow is an extreme point of $P_F$ if and only if it is a cycle free solution.
\end{proposition}

\begin{proposition} \label{prop:tree}
Given an extreme point of $P_F$, we can construct a spanning tree solution associated with this extreme point, and the resulting spanning tree solution may not be unique.
\end{proposition}
\begin{proof}
Consider an extreme $f$ in graph $G=(V,A)$. From Proposition~\ref{prop:cyclefree}, $f$ is a cycle free solution. Let $A_f$ be the set of free arcs, then $(V, A_f)$ is a forest. If the forest is a spanning tree, then $f$ and this spanning tree forms a spanning tree solution. Otherwise, we add restricted arcs not in $A_f$ into the forest one at a time, in a way that no undirected cycle is created. In the end, a spanning tree will be produced, and $f$ with the produced spanning tree gives a spanning tree solution. When $(V,A_f)$ is not a spanning tree, the produced spanning tree may contain different restricted arcs, so the spanning tree solution is not unique.
\end{proof}

We now give a formula to compute the flow over any arc in an extreme point $f$. We first create a spanning tree solution based on Proposition~\ref{prop:tree}. Call the associated spanning tree $\bT_f$. Consider an arbitrary arc $a=(u, v)$ in $G$. If $a$ is not in $\bT_f$, then the flow over $a$ is zero or $U_a$. If $a$ is in $\bT_f$, removing $a$ from $\bT_f$ breaks $\bT_f$ into two connected components. Call the two components sub-trees $\bT_{f,1}$ and $\bT_{f,2}$, respectively. Without loss of generality, assume that one endpoint $u$ is in $\bT_{f,1}$ and another endpoint $v$ is in $\bT_{f,2}$. By the flow balance constraints, we have 
\begin{equation}\label{eq:flow:a}
f_a = \sum_{v_{l, t} \in \bT_{f,1}}b_{l, t} + \sum_{e\in A_2} f_e - \sum_{e \in A_1}f_e,
\end{equation}
where $A_1$ is the set of nontree arcs with tails in $\bT_{f,1}$ and heads in $\bT_{f,2}$, and $A_2$ is the set of nontree arcs with tails in $\bT_{f,2}$ and heads in $\bT_{f,1}$. 

In the rest of the paper, our main focus is to use~\eqref{eq:flow:a} to enumerate values of $f_a$ in all extreme points $f$ of $P_F$. We will present the analysis for the C-MFG first, and then for the U-MFG. We introduce a crucial lemma that will be used in both cases. Below a ``path'' could refer to either an undirected path or a directed path, when the context is clear. We say two paths are \emph{vertex-disjoint} if the two paths have no vertex in common. The lemma relies on the planar embedding of the grid $G$.

\begin{lemma} \label{lem:intersect}
Let $v^1,v^2,v^3$ and $v^4$ be four distinct vertices lying clockwise on the boundary of the outer face of a plane graph. Suppose $P_1$ is a path between $v^1$ and $v^3$ and $P_2$ is a path between $v^2$ and $v^4$. Then $P_1$ and $P_2$ cannot be vertex-disjoint.
\end{lemma}

\section{The C-MFG} \label{sec:CFG}
We first present two set of conditions for the C-MFG to be polynomially solvable, and then complement that with several NP-hard cases when any of the condition is removed.
\subsection{Polynomially solvable cases}
\begin{theorem} \label{thm:ccfgl_p}
The C-MFG can be solved in $O(L^{4KL-4K+1}T^{4KL+4L-4K-3})$ time, if sources and sinks are in at most two rows.
\end{theorem}
\begin{proof}
First we assume sources and sinks are in row one and row $L$, i.e., all sources and sinks are on the boundary of the outer face of $G$. This is without loss of generality. If the sources and sinks are in row $l_1$ and row $l_2$, since all arcs in $G$ are forward and downward arcs, there will be zero flow in each of the arcs above row $l_1$ and below row $l_2$. We can eliminate these arcs from the grid and solve the problem over a smaller grid.

Given an arc $a$, an extreme point $f$, and its associated spanning tree, 
\[f_a = \sum_{v_{l, t} \in \bT_{f,1}}b_{l, t} + \sum_{e\in A_2} f_e - \sum_{e \in A_1}f_e,\]
by equation~\eqref{eq:flow:a}.  Our goal is to investigate the values of $\sum_{v_{l, t} \in \bT_{f,1}} b_{l, t}$ and $\sum_{e\in A_2} f_e - \sum_{e \in A_1}f_e$ in all extremes points, respectively. 

Consider the term $\sum_{v_{l, t} \in \bT_{f,1}}b_{l, t}$ first. By our assumption, all sources and sinks are on the boundary of the outer face of $G$. We claim that all sources and sinks in $\bT_{f,1}$ appear consecutively. In other words, there exist two vertices $v'$ and $v''$ in $\bT_{f,1}$ such that, if we walk from $v'$ to $v''$ along the boundary of the outer face of $G$, then all sources and sinks in $\bT_{f,1}$ will be visited with no source or sink in $\bT_{f,2}$ on the way. Thus the possible values for $\sum_{v_{l, t} \in \bT_{f,1}}b_{l, t}$ in all extreme points is contained in the set
\[S=\left\{\sum_{t=i}^j b_{1,t}, \sum_{t=i}^j b_{L,t}, \sum_{t=1}^ib_{1,t} + \sum_{t=1}^j b_{L,t}, \sum_{t=i}^Tb_{1,t} + \sum_{t=j}^T b_{L,t}, \text{ for } i,j\in [T]\right\}.\]
The set $S$ can be constructed in $O(T^2)$ time, with cardinality $O(T^2)$. To prove the claim, suppose there exist four vertices $v^i$ ($i=1,\ldots,4$) lying clockwise on the boundary of the outer face of $G$, with $v^1,v^3\in \bT_{f,1}$ and $v^2,v^4\in \bT_{f,2}$. Since $\bT_{f,1}$ is connected, there is a path between $v^1$ and $v^3$. Similarly there is a path between $v^2$ and $v^4$ in $\bT_{f,2}$. By Lemma~\ref{lem:intersect}, the two paths are not vertex-disjoint. But that implies $\bT_1$ and $\bT_2$ are connected, a contradiction. 

For the term $\sum_{e\in A_2} f_e - \sum_{e \in A_1}f_e$, each arc $e$ in $A_1\cup A_2$ is a nontree arc, so $f_e=0$ or $f_e=U_e$. Suppose there are $K$ different arc capacity values in $G$. With $O(LT)$ arcs in $G$, the maximum number of different values for $\sum_{e\in A_2} f_e - \sum_{e \in A_1}f_e$ is $O((LT)^{2K})=O((LT)^{2K})$, and these values can be enumerated in $O((LT)^{2K})$ time as well.

Therefore, the number of different values for $f_a$ under all extreme points is $O(T^2)\cdot O((LT)^{2K}) = O(L^{2K} T^{2K+2})$. From Theorem~\ref{prop:runtime}, the C-MFG can be solved in $O(L^{4KL-4K+1}T^{4KL+4L-4K-3})$ time.
\end{proof}

\begin{corollary} \label{cor:ccfg}
The C-MFG with constant $L$ and $K$ and sources and sinks in at most two rows is polynomially solvable.
\end{corollary}

If the grid has only two rows (the case for the basic capacitated lot sizing model, studied in~\cite{florian1971deterministic,van1996t3}), the condition in Corollary~\ref{cor:ccfg} can be relaxed. Instead of requiring the number of different capacity values over \emph{all} arcs to be constant, we
allow the forward arcs to have arbitrary capacities, and require the number of different capacity values over \emph{all downward arcs} to be constant. This condition cannot be further relaxed, since we know the lot sizing problem with general production capacity is NP-hard. Define
\[K_1 =\{U_a \mid U_a < +\infty, a=(v_{1,t},v_{2,t}) \text{ for some } t\in [T]\}\]
to be the number of different capacity values over downward arcs.

\begin{theorem} \label{prop:CFG-2CDpoly}
The C-MFG with two rows and $K_1$ different capacity values on downward arcs can be solved in $O(T^{4K_1+7})$ time; the problem is polynomially solvable with constant $K_1$.
\end{theorem}
\begin{proof}
The idea of proof is similar to that of Theorem~\ref{thm:ccfgl_p}. Given a forward arc $a=(u,v)$, an extreme point $f$, and its associated spanning tree, 
\[f_a = \sum_{v_{l, t} \in \bT_{f,1}}b_{l, t} + \sum_{e\in A_2} f_e - \sum_{e \in A_1}f_e,\]
by equation~\eqref{eq:flow:a}. By a similar argument as in Theorem~\ref{thm:ccfgl_p}, the term $ \sum_{v_{l, t} \in \bT_{f,1}}b_{l, t}$ can take $O(T^2)$ different values. Consider the term $\sum_{e\in A_2} f_e - \sum_{e \in A_1}f_e$. Since the grid only has two rows, the set $A_1\cup A_2$ can contains at most one forward nontree arc between $\bT_{f,1}$ and $\bT_{f,2}$ after arc $a$ is deleted, and the flow over that forward arc can take $O(T)$ different values; the set $A_1 \cup  A_2$ contains at most $T$ downward arcs, each of which can take $K_1$ different values. Thus $\sum_{e\in A_2} f_e - \sum_{e \in A_1}f_e$ can take $O(T^{2K_1+1})$ values. Then $f_a$ can take $O(T^{2K_1+3})$ values under all extreme points of $P_F$. From Theorem~\ref{prop:runtime}, the C-MFG with $L=2$ and $K_1$ different capacity values on downward arcs can be solved in $O(T^{4K_1+7})$ time.
\end{proof}

\subsection{NP-hard cases} \label{sec:ccfgnphard}
In this section, we will show that the problems in Corollary~\ref{cor:ccfg} and Theorem~\ref{prop:CFG-2CDpoly} are essentially the most general polynomially solvable C-MFG cases unless P=NP. Before stating any results, we want to point out that all proofs for the NP-hard cases in rest of the paper will be accompanied by a figure, to better illustrate the reduction from an instance of an NP-hard problem to an MFG instance. Legends in these figures describe the parameters of the MFG instance: the amount of supply or demand is marked next to the corresponding vertex; the pair $(c,U)$ next to an arc $a$ denotes that the cost of sending any nonzero flow over arc $a$ is $c$ and 0 otherwise, and the capacity of arc $a$ is $U$; an arc without such parameters next to it has infinity capacity and zero cost of sending any flow; an arc not present in the figure denotes that it will never be used in any optimal solution because of its large cost. See Figure~\ref{fig:ccfg2sinks} below for an example.

The next proposition shows the C-MFG becomes NP-hard if sources and sinks are in three rows.
\begin{proposition} \label{thm:CCFGNP-hard}
The C-MFG with $L=3$, $K=1$, a single source, and sinks in two other rows is NP-hard.
\end{proposition}
\begin{proof}
Our proof is based on a reduction from the knapsack problem~\cite{garey1979computers}. The knapsack problem asks that given a set of $n$ items with item $i$ having value $y_i$ and cost $c_i$ for $i=1,\ldots,n$, if there exists a subset of items with cost at most $C$ and total value at least $Y$. We construct a C-MFG instance with $L=3$, $K=1$, a single source, and sinks in two rows as follows. First choose a value $B\ge\max\{y_i \mid i \in [n]\}$. Construct a C-MFG instance with three rows, $n$ columns, one source ($v_{1,1}$ with supply $\sum_{i=1}^n(B-y_i)+Y$), and $n+1$ sinks ($v_{2,i}$ with demand $(B-y_i)$ for $i\in [n]$ and $v_{3,n}$ with demand $Y$), as shown in Figure~\ref{fig:ccfg2sinks}. The cost over each arc $(v_{2,i},v_{3,i})$ has a fixed cost $c_i$ for nonzero flow and 0 otherwise, the cost over each of the rest of arcs in Figure~\ref{fig:ccfg2sinks} is always 0, and the cost over each of the arcs not present in Figure~\ref{fig:ccfg2sinks} is always large enough, say $\sum_{i=1}^nc_i + C$, so that they are never used in any optimal solution. The capacity of each downward arc is $B$, and the capacity of each forward arc is $\infty$.

\begin{figure}[htb]
\centering
\includegraphics[scale=0.75]{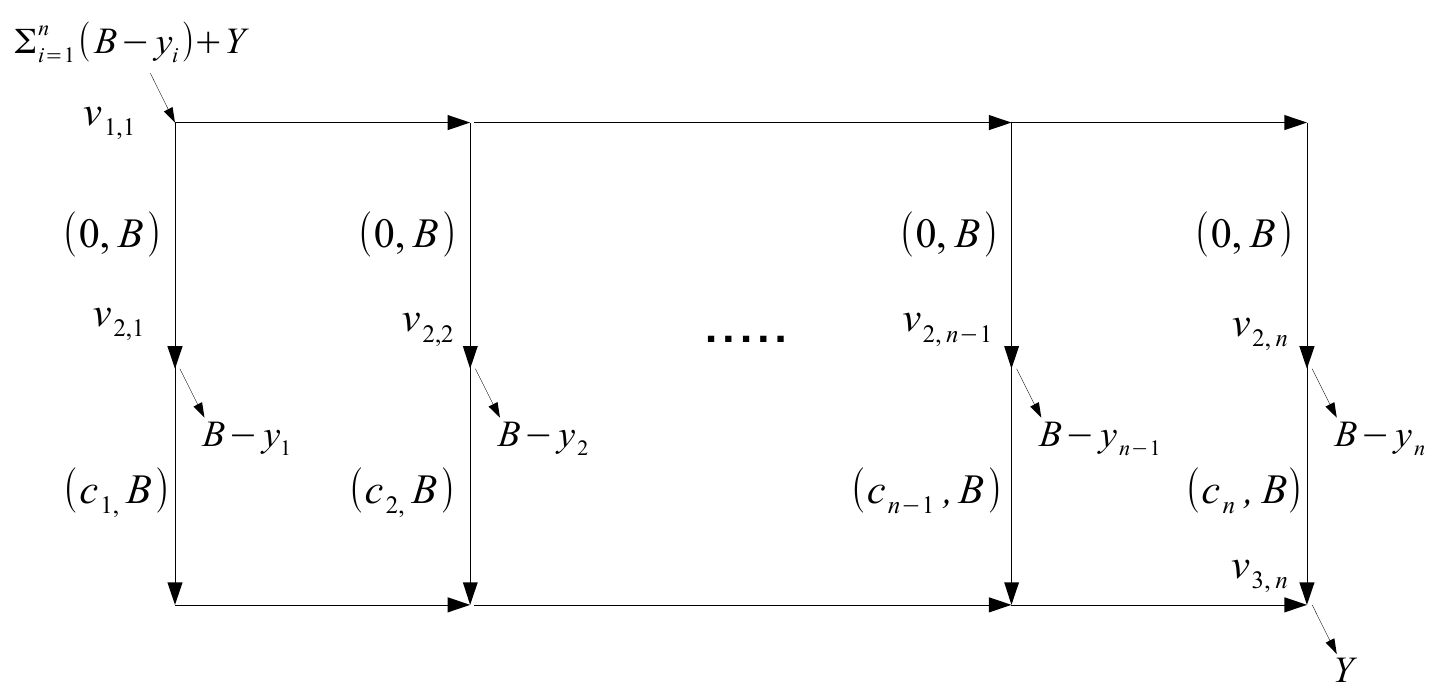}
\caption{The C-MFG with $L=3$, $K=1$, a single source, and sinks in two rows} 
\label{fig:ccfg2sinks}
\end{figure}

We claim that the knapsack instance is a yes instance if and only if the optimal objective value of the constructed C-MFG instance is no greater than $C$. We first show that given any feasible solution of the knapsack instance, we can construct a feasible flow in the C-MFG instance with objective at most $C$. Let $\epsilon = \sum_{i=1}^n y_i -Y \ge 0$. Since $\epsilon \le \sum_{i=1}^n y_i$, we can find $\epsilon_1,\ldots, \epsilon_n$ such that $0\le \epsilon_i < y_i$ for $i\in [n]$ and $\sum_{i=1}^n \epsilon_i = \epsilon$. Given a feasible solution $S\subseteq [n]$ of the knapsack instance with cost $\sum_{i\in S}c_i$, construct a feasible flow in the C-MFG instance as follows: the flow sent over arc $(v_{1,i},v_{2,i})$ for $i \notin S$ equals the demand $(B-y_i)$ at $v_{2,i}$; the flow sent over arc $(v_{1,i},v_{2,i})$ for $i \in S$ equals $B - \epsilon_i$; the flow sent over arc $(v_{2,i},v_{3,i})$ for $i\in S$ equals $y_i - \epsilon_i$; the flow sent over other arcs are calculated according to the flow conservation constraints. The cost of the constructed feasible flow is $\sum_{i\in S} c_i$, which is at most $C$. On the other hand, if the optimal objective value of the C-MFG instance is at most $C$, then the knapsack instance must be a yes instance. Suppose an optimal flow $f^*$ in the C-MFG instance has cost no greater than $C$. Let $f^*_i$ be the flow over arc $(v_{2,i},v_{3,i})$ for $i\in [n]$. Let $S=\{i \mid f^*_i >0\}$. We claim that $S$ is a feasible solution of the knapsack instance, i.e., $\sum_{i\in S}c_i \le C$ and $\sum_{i\in S}y_i \ge Y$. The cost of flow $f^*$ is $\sum_{i\in S} c_i$, which is no greater than $C$. Since the capacity of downward arc $(v_{1,i}, v_{2,i})$ is $B$ and the demand at vertex $v_{2,i}$ is $B-y_i$, the maximum amount of flow that can be sent along downward arc $(v_{2,i},v_{3,i})$ is at most $B-(B-y_i)=y_i$ for $i\in S$. Thus $\sum_{i\in S}y_i \ge \sum_{i\in S}f^*_i = Y$.
\end{proof}

The capacitated lot sizing problem indicates that the C-MFG with arbitrary capacities on downward arcs is NP-hard. The following proposition shows that the C-MFG with arbitrary capacities on forward arcs is also NP-hard. Both imply the condition that the number of different capacity values $K$ being constant is crucial for the C-MFG to be polynomially solvable.
\begin{proposition} \label{cor:ccfg3}
The C-MFG with $L=3$, a single source, a single sink, and arbitrary capacities on forward arcs is NP-hard.
\end{proposition}
\begin{proof}
Our proof is based on a reduction from the knapsack problem. We construct a C-MFG instance with three rows, a single source $v_{1,1}$ with supply $Y$, and a single sink $v_{3,2n}$ with demand $Y$, as shown in Figure~\ref{fig:ccfg_horizon}. The forward arc $(v_{2,2i-1},v_{2,2i})$ has a capacity $y_i$, and a cost of $c_i$ of sending any nonzero flow and 0 otherwise, for $i\in [n]$. Each of the remaining arcs in Figure~\ref{fig:ccfg_horizon} has zero cost of sending any flow and an infinite capacity. 

\begin{figure}[htb]
\centering
\includegraphics[scale=0.7]{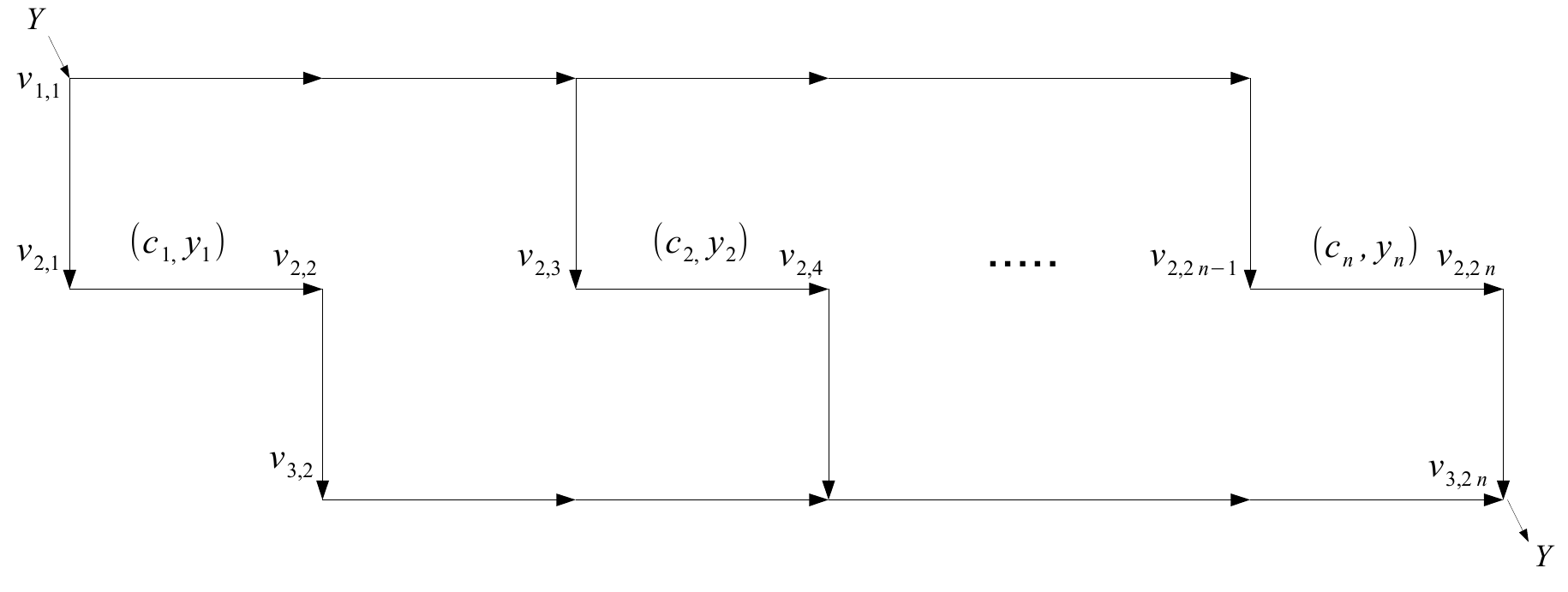}
\caption{The C-MFG with three rows, a single source, a single sink, and general capacity on forward arcs} 
\label{fig:ccfg_horizon}
\end{figure}

We claim that the knapsack instance is a yes instance if and only if the optimal objective value of the constructed C-MFG instance is no greater than $C$. Let $S\subseteq [n]$ be a feasible solution of the knapsack instance, i.e., $\sum_{i\in S}c_i \le C$ and $\sum_{i\in S}y_i \ge Y$. Let $\epsilon = \sum_{i=1}^n y_i -Y \ge 0$. Since $\epsilon \le \sum_{i=1}^n y_i$, we can find $\epsilon_1,\ldots, \epsilon_n$ such that $0\le \epsilon_i < y_i$ for $i=1,\ldots,n$ and $\sum_{i=1}^n \epsilon_i = \epsilon$. Construct a feasible flow in the C-MFG instance as follows: the flow sent over arc $(v_{2, 2i-1}, v_{2,2i})$ is $y_i -\epsilon_i$ for $i\in S$; the flow sent over other arcs are calculated according to the flow conservation constraints. The cost of the constructed feasible flow is $\sum_{i\in S}c_i$, which is at most $C$. On the other hand, let $f^*$ be an optimal flow of the C-MFG instance and $f^*_i$ be the flow over arc $(v_{2,2i-1}, v_{2,2i})$ for $i \in [n]$. Let $S=\{i \mid f^*_i >0\}$. Then it is easy to verify that $S$ is a feasible solution of the knapsack instance.
\end{proof}

\begin{remark}
The conditions for a capacitated lot sizing model to be polynomially solvable are very subtle. By Proposition~\ref{thm:CCFGNP-hard}, the \emph{constant capacitated} multi-echelon lot sizing problem with intermediate demands is NP-hard, in contrast to the polynomially solvable \emph{uncapacitated} multi-echelon lot sizing problem with intermediate demands~\cite{zhang2011polyhedral,he2012minimum}. By Proposition~\ref{cor:ccfg3}, the uncapacitated \emph{multi-echelon} lot sizing problem with variable storage bounds is NP-hard, in contrast to the polynomially solvable \emph{single-echelon} lot sizing problem with variable storage bounds~\cite{atamturk2008n2}.
\end{remark}

\section{The U-MFG} \label{sec:UCFG}
We first present a general condition for the U-MFG to be polynomially solvable, and then complement that with several NP-hard cases.
\subsection{Polynomially solvable cases}
\begin{theorem} \label{thm:CFGL} 
The U-MFG with sources in one row can be solved in $O(LT^{8L^2-12L+5})$ time.
\end{theorem}

\noindent Theorem~\ref{thm:CFGL} immediately implies the following result.
\begin{corollary} \label{cor:mfgl}
The U-MFG with sources in one row and a constant number of rows is polynomially solvable.
\end{corollary}

\noindent The main ingredient to prove Theorem~\ref{thm:CFGL} is to enumerate all possible flow values over any arc in all extreme points of $P_F$. This is summarized in the following proposition.

\begin{proposition} \label{prop:CFGLInflow}
In the U-MFG with sources in one row, the flow over any arc can take $O(T^{4L-2})$ different values in all extreme points of $P_F$, and these values can be enumerated in $O(T^{4L-2})$ time.
\end{proposition}
\noindent Theorem~\ref{thm:CFGL} then follows immediately from Theorem~\ref{prop:runtime} and Proposition~\ref{prop:CFGLInflow}. The proof of Proposition~\ref{prop:CFGLInflow} is slightly technical and requires several new gadgets, which we introduce below. The proof can be summarized in two steps: (1) For each extreme point of $P_F$, construct a corresponding spanning tree solution with certain property; (2) With the constructed spanning tree solution, the flow over any arc can be written as a sum of supplies at vertices following a special pattern. Then we can enumerate the possible flow values in all extreme points in an efficient way, instead of simply enumerating all extreme points and calculating the flow values accordingly.

\subsubsection{A special spanning tree solution}
We first assume that all sources are in row one. Otherwise, we can remove all arcs in rows above the sources, without changing the optimal flow over the remaining arcs. We need the following concept before constructing the spanning tree with the required property.
\begin{definition}
Given a tree of the grid, a vertex is said to be \emph{accessible from row one}, if there exists a directed path in the tree starting from some vertex in row one to that vertex.
\end{definition}
\noindent Consider the spanning tree in Figure~\ref{fig:ucfgL}, the vertex $v_{l,i_1}$ is accessible from row one, since there is a directed path from a first-row vertex $v_{1,i_1}$ to $v_{l,i_1}$. In fact, each vertex in that tree is accessible from row one. We will construct a spanning tree with this property for every extreme point of $P_F$. In particular, given an extreme point $f$, let $A_f=\{ e \mid f_e >0\}$. By Proposition~\ref{prop:cyclefree}, $(V, A_f)$ is a forest. Our goal is to (possibly) include some restricted arcs not in $A_f$ to produce a spanning tree $\bT_f$ with the following property.
\begin{properties}
\item Every vertex $v$ in $G$ is accessible from row one. \label{property:accessible}
\end{properties}
This can be achieved by Algorithm~\ref{alg:spanning}.

\begin{algorithm} [htb]
\caption{Construct a spanning tree solution satisfying Property~\ref{property:accessible}}
\label{alg:spanning}
\begin{algorithmic}
\Require An extreme point $f$ of $P_F$
\Ensure A spanning tree $\bT_f = (V, A_{\bT_f})$ satisfying Property~\ref{property:accessible}
\State Initialization: $A_{\bT_f} \gets A_f$
\While{$(V,A_{\bT_f})$ is not a spanning tree}
\State Select an arc $e=(u_e, v_e) \notin A_{\bT_f}$ such that $u_e$ is accessible from row one and $(V, A_{\bT_f}\cup \{e\})$ does not contain an undirected cycle.
\State $A_{\bT_f} \gets A_{\bT_f}\cup \{e\}$
\EndWhile
\end{algorithmic}
\end{algorithm}

\begin{proposition}
The spanning tree constructed by Algorithm~\ref{alg:spanning} satisfies Property~\ref{property:accessible}.
\end{proposition}
\begin{proof}
During the construction of the tree, we call a vertex in $G$ \emph{isolated} if the vertex is not adjacent to any arc in $A_{\bT_f}$ and \emph{non-isolated} otherwise. We prove the result by induction on the number of non-isolated vertices. During the initialization step, all non-isolated vertices, i.e., vertices incident to at least one arc in $A_{f}$, are accessible from row one, due to the flow conservation constraints and the fact that all sources are in row one. At each iteration, if there exists an arc $e=(u_e,v_e)$ not in $A_{\bT_f}$ such that $u_e$ is non-isolated and $v_e$ is isolated, then add arc $e$ into $A_{\bT_f}$. Since $u_e$ is accessible from row one, vertex $v_e$ becomes non-isolated and is accessible from row one as well. If such an arc does not exist, then there must exist an isolated vertex $v$ that is in row one. Select one of the two outgoing arcs of vertex $v$ as arc $e$, and include it into $A_{\bT_f}$. Then $v$ is accessible from row one by definition.
\end{proof}

The constructed spanning tree has some additional nice properties. To establish those properties, we first define an auxiliary function $\kappa_l: [T] \rightarrow [T] $ for each $l\in [L]$.
\begin{definition} Given a spanning tree $\bT_f$ with Property~\ref{property:accessible}, for each $l\in [L]$,
\begin{enumerate}
	\item the value $\kappa_l(i)$ is the smallest integer $j$ such that there is a directed path from $v_{1, j}$ to $v_{l, i}$ in $\bT_f$.
	\item the path $P_{l,i}$ denotes the unique directed path in $\bT_f$ from $v_{1, \kappa_l(i)}$ to $v_{l, i}$.
\end{enumerate} 
\end{definition}

\begin{figure}[htb]
    \centering
    \includegraphics[scale=0.65]{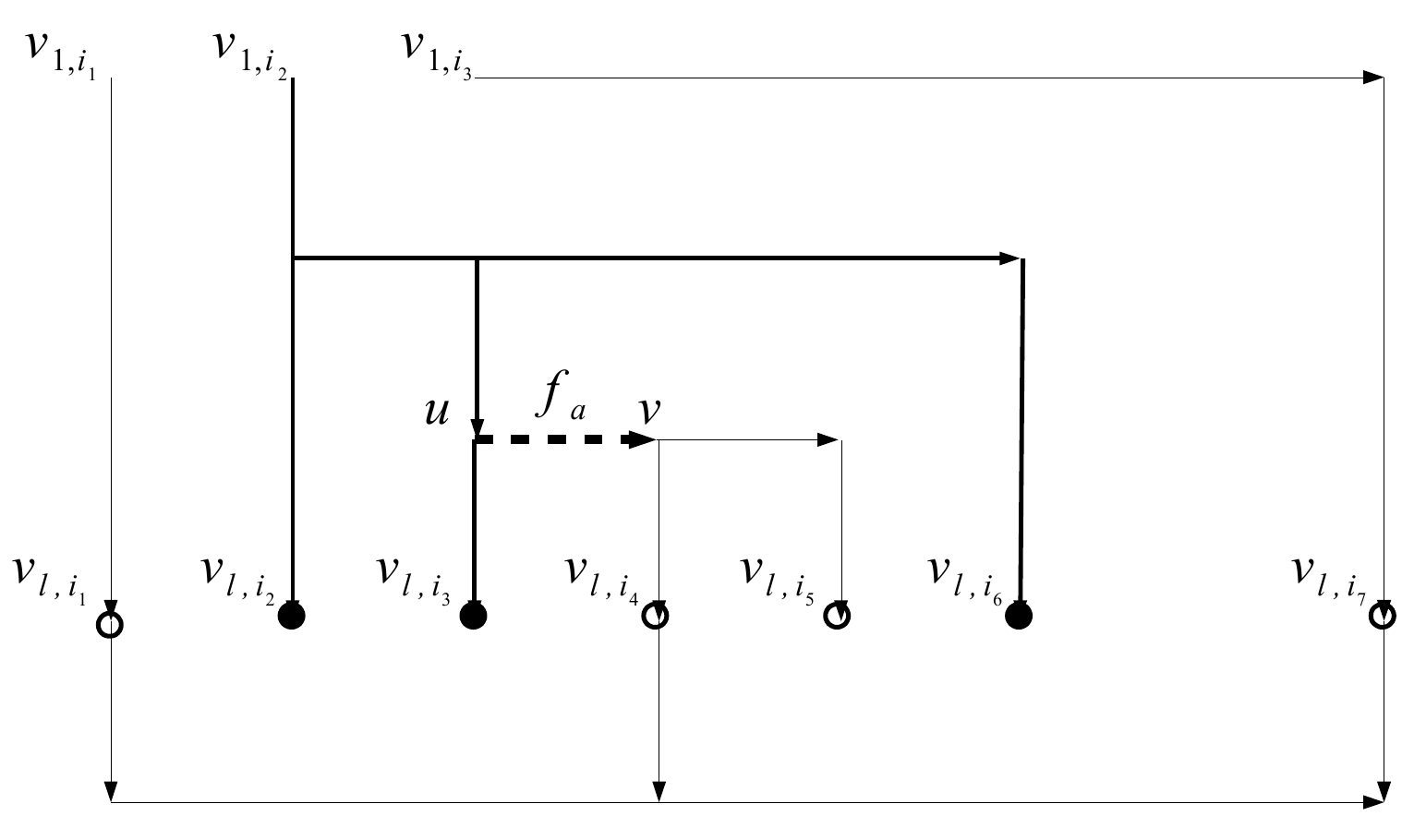}        
    \caption{The spanning tree $\bT_f$ satisfying Property~\ref{property:accessible}, the given arc $a=(u,v)$, and vertices of type 1 (hollow dots) and type 2 (solid dots) in row $l$}
    \label{fig:ucfgL}
\end{figure}

\noindent For the example in Figure~\ref{fig:ucfgL}, $\kappa_l(i_1)=i_1$, $\kappa_l(i_2)=\ldots =\kappa_l(i_6)=i_2$, and $\kappa_l(i_7)=i_3$. The path $P_{l,i_1}$ consists of downward arcs $(v_{1,i_1}, v_{2,i_1}), \ldots, (v_{l-1,i_1}, v_{l,i_1})$, and the path $P_{l,i_7}$ consists of forward arcs $(v_{1,i_3},v_{1,i_4}), \ldots, (v_{1,i_6},v_{1,i_7})$ and downward arcs $(v_{1,i_7},v_{2,i_7}), \ldots, (v_{l-1,i_7},v_{l,i_7})$.

\begin{lemma}
\label{lemma:kappa} $\;$
\begin{enumerate}
	\item For each $l\in [L]$, the function $\kappa_l$ is non-decreasing.
	\item The union of paths $\set{P_{l,i}: i \in [T]}$ forms a forest of arborescences. 
\end{enumerate} 
\end{lemma}
\begin{proof}
Fix $l\in [L]$. Suppose for some $i < i'$ we have $\kappa_l(i) = j > \kappa_l(i') = j'$. Consider the two directed paths $P_{l,i}$ and $P_{l,i'}$ in $\bT_f$. We claim they cannot be vertex-disjoint. To see this,  first the two paths do not contain any arcs below row $l$, since all arcs are either forward or downward in $G$. Consider the subgraph of $G$ consisting of only vertices and arcs in row one to row $l$. It is also planar, and vertices $v_{1,j'}, v_{1,j}, v_{l,i'}$, and $v_{l,i}$ lie clockwise on the boundary of its outer face. By Lemma~\ref{lem:intersect}, the two paths $P_{l,i}$ and $P_{l,i'}$ cannot be vertex-disjoint. Then there will be a directed path from $v_{1,j'}$ to $v_{l, i}$ in $\bT_f$, contradicting the definition of $\kappa_l(i)$. 

To prove the second statement, it suffices to show that if $P_{l,i}$ and $P_{l,i'}$ are not vertex-disjoint, then they must start from the same vertex in the first row. Suppose they start from two different vertices, $v_{1, j}$ and $v_{1,j'}$ with $j < j'$. Since the two paths are not vertex-disjoint, there is a directed path from $v_{1,j}$ to $v_{l, i'}$. The path $P_{l,i'}$ should start from $v_{1,j}$ instead of $v_{1,j'}$, a contradiction. 
\end{proof}
\noindent Consider the example in Figure~\ref{fig:ucfgL}. For vertices in row $l$, $\kappa_l(i_1)=i_1$, $\kappa_l(i_2)=\ldots =\kappa_l(i_6)=i_2$, and $\kappa_l(i_7)=i_3$, so $\kappa_l$ is non-decreasing; the union of paths $P_{l,i}$ ($i\in [T]$) forms three arborescences: the directed path from $v_{1,i_1}$ to $v_{l,i_1}$, the arborescence rooted at $v_{1,i_2}$ with leaves $v_{l,i_2}$ to $v_{l,i_6}$, and the directed path from $v_{1,i_3}$ to $v_{l,i_7}$.

\subsubsection{The flow value $f_a$}
It is known that if we delete an arc $a=(u,v)$ from a spanning tree $\bT_f$, the tree will break into two subtrees $\bT_{f,1}$ (containing vertex $u$) and $\bT_{f,2}$ (containing vertex $v$). We classify all vertices in $V$ according to their locations in the subtrees.
\begin{definition}
Given an arc $a=(u,v)\in A$ and a spanning tree $\bT_f$ with Property~\ref{property:accessible},
\begin{enumerate}
	\item a vertex is of \emph{type 1} if it is contained in $\bT_{f,1}$, and of \emph{type 2} if it is contained in $\bT_{f,2}$;
	\item a type 2 vertex is of \emph{type 2A}, if the directed path in $\bT_f$ from the vertex in row one to that vertex contains arc $a$; a type 2 vertex is of \emph{type 2B} otherwise.
\end{enumerate}
\end{definition}

\noindent Consider the example in Figure~\ref{fig:ucfgL} with arc $a= (u,v)$. For vertices in row one, vertices $v_{1,i_1}$ and $v_{1,i_3}$ are of type 2, and vertex $v_{1,i_2}$ is of type 1. For vertices in row $l$, vertices $v_{l,i_2}, v_{l,i_3}$, and $v_{l,i_6}$ are of type 1, $v_{l,i_4}$ and $v_{l,i_5}$ are of type 2A, and $v_{l,i_1}$ and $v_{l,i_7}$ are of type 2B.

We introduce several lemmas to characterize the location of vertices of each type in a given row. These lemmas will be used later for enumerating the values of $f_a$ in all extreme points.
\begin{lemma}
\label{lemma:row1}
Given an arc $a=(u,v)\in A$ and a spanning tree $\bT_f$ with Property~\ref{property:accessible}, there cannot be four vertices in row one that alternate between type 1 and type 2.
\end{lemma}
\begin{proof}
Suppose there exist $v_{1,i}, v_{1, j}, v_{1, i'}$, and $v_{1, j'}$ with $i < j < i' < j'$ such that $v_{1,i}$ and $v_{1,i'}$ are of the same type and $v_{1,j}$ and $v_{1,j'}$ are of another type. WLOG, assume $v_{1,i}$ and $v_{1,i'}$ are of type 1 and $v_{1,j}$ and $v_{1,j'}$ are of type 2. Then there is an undirected path between $v_{1,i}$ and $v_{1,i'}$ in $\bT_{f,1}$ and an undirected path between $v_{1,j}$ and $v_{1,j'}$ in $\bT_{f,2}$. These two paths cannot be vertex-disjoint by Lemma~\ref{lem:intersect}. That makes $\bT_{f,1}$ and $\bT_{f,2}$ connected, a contradiction.
\end{proof}

\begin{lemma}
\label{lemma:2A}
Given an arc $a=(u,v)\in A$ and a spanning tree $\bT_f$ with Property~\ref{property:accessible}, if $v_{l,i}$ and $v_{l,j}$ are of type 2A with $i<j$, then $v_{l, k}$ is also of type 2A, for any $k$ such that $i<k<j$ and each $l \in [L]$. 
\end{lemma}
\begin{proof}
Consider the directed paths $P_{l,i}$ and $P_{l,j}$ that go from some vertices in row one to $v_{l,i}$ and $v_{l,j}$ in $\bT_f$, respectively. By the definition of type 2A vertices, both $P_{l,i}$ and $P_{l,j}$ contain arc $a$. Then $\kappa_l(i) = \kappa_l(j)$, following from the second statement of Lemma~\ref{lemma:kappa}. Since $i<k<j$, $\kappa_l(i) \le \kappa_l(k) \le \kappa_l(j)$ by the first statement of Lemma~\ref{lemma:kappa}. Then $\kappa_l(i) =\kappa_l(k) = \kappa_l(j)$. The directed paths that connect vertices in row one to $v_{l,i}$, $v_{l,k}$, and $v_{l,j}$ start from the same vertex $v_{1,\kappa_l(i)}$ in row one. Since the paths $P_{l,i}$ and $P_{l,j}$ both contain arc $a$ and there is only one path in $\bT_{f}$ from vertex $v_{1,\kappa_l(i)}$ to vertex $v$ (the head of $a$), the segment in $P_{l,i}$ from $v_{1,\kappa_l(i)}$ to $v$ should be the same as the segment in $P_{l,j}$ from $v_{1,\kappa_l(i)}$ to $v$, which is the same as the segment in $P_{l,k}$ from $v_{1,\kappa(i)}$ to $v$ as well. Thus $P_{l,k}$ contains arc $a$ and $v_{l,k}$ is also of type 2A.
\end{proof}

\begin{lemma}
\label{lemma:rowl}
Given an arc $a=(u,v)\in A$ and a spanning tree $\bT_f$ with Property~\ref{property:accessible}, there cannot be four vertices in row $l $ ($l>1$) that alternate between type 1 and type 2B.
\end{lemma}
\begin{proof}
Suppose there exist $v_{l,i}, v_{l, j}, v_{l, i'}$, and $v_{l, j'}$ with $i < j < i' < j'$ such that $v_{l,i}$ and $v_{l,i'}$ are of the same type and $v_{l,j}$ and $v_{l,j'}$ are of another type. WLOG, assume $v_{l,i}$ and $v_{l,i'}$ are of type 1 and $v_{l,j}$ and $v_{l,j'}$ are of type 2B. Consider the four directed paths $P_{l,i}, P_{l,j}, P_{l,i'}, P_{l,j'}$ in $\bT_f$. None of these paths contains arc $a$ according to the definitions of type 1 and type 2B vertices, so $P_{l,i}$ and $P_{l,i'}$ are contained entirely in $\bT_{f,1}$, and $P_{l,j}$ and $P_{l,j'}$ are contained entirely in $\bT_{f,2}$. Consider the starting vertices of the four directed paths: $v_{1, \kappa_l(i)}$ and $v_{1, \kappa_l(j)}$ are of type 1, and $v_{1, \kappa_l(i')}$ and $v_{1, \kappa_l(j')}$ are of type 2; moreover, $\kappa_l(i) < \kappa_l(j) < \kappa_l(i') < \kappa_l(j')$ by Lemma~\ref{lemma:kappa}. Thus we found four vertices in row one that alternate between type 1 and type 2, which contradicts Lemma~\ref{lemma:row1}.
\end{proof}

Given an arc $a=(u,v)\in A$ and a spanning tree $\bT_f$ with Property~\ref{property:accessible}, the vertices in row $l$ for each $l\in [L]$ can be partitioned into groups of consecutive type 1 and type 2 vertices. 
\begin{definition}
Given an arc $a=(u,v)\in A$ and a spanning tree $\bT_f$ with Property~\ref{property:accessible}, an \emph{interval of type 1 (type 2, type 2A, type 2B)} in row $l$ ($l\in [L]$) is a set of adjacent vertices $\{v_{l,i} \mid i_1 \le i \le j_1 \text{ for some } i_1, j_1\}$ of type 1 (type 2, type 2A, type 2B). An interval is \emph{maximal} if it is not contained in any larger interval of the same type. 
\end{definition}
\noindent Consider the example in Figure~\ref{fig:ucfgL}. Vertices $v_{l,i_2}$ and $v_{l,i_3}$ form a maximal interval of type 1, and vertices $v_{l,i_4}$ and $v_{l,i_5}$ form a maximal interval of type 2A. The following lemma shows that the number of maximal intervals in each row is small.
\begin{lemma}
\label{lemma:intervals}
Given an arc $a=(u,v)\in A$ and a spanning tree $\bT_f$ with Property~\ref{property:accessible}, row $l$ has at most two maximal intervals of type 2, or at most two maximal intervals of type 1, for each $l\in [L]$.
\end{lemma}
\begin{proof}
When $l=1$, there are at most two maximal intervals of type 1 and two maximal intervals of type 2, according to Lemma~\ref{lemma:row1}. For $l>1$, we consider two different cases, depending on whether row $l$ contains any type 2A vertex or not. If row $l$ does not contain any vertex of type 2A, then it can contain at most two maximal intervals of type 2B, according to Lemma~\ref{lemma:rowl}. If row $l$ contains at least one type 2A vertex, then all type 2A vertices form a maximal interval of type 2A, according to Lemma~\ref{lemma:2A}. Suppose there are at least three maximal intervals of type 1 and three maximal intervals of type 2 in row $l$. We already know one maximal interval of type 2 contains all type 2A vertices in that row. For the remaining five intervals, we can always find four vertices, $v_{l,i}, v_{l, j}, v_{l, i'}$, and $v_{l, j'}$ with $i < j < i' < j'$, that alternate between type 1 and type 2B, which contradicts Lemma~\ref{lemma:rowl}.
\end{proof}
\noindent Consider the example in Figure~\ref{fig:ucfgL}. Vertices in row $l$ present in Figure~\ref{fig:ucfgL} are divided into two types: $v_{l,i_2},v_{l,i_3}$, and $v_{l,i_6}$ are of type 1, and vertices $v_{l,i_1},v_{l,i_4},v_{l,i_5}$, and $v_{l,i_7}$ are of type 2. Row $l$ contains only two maximal intervals of type 1.

Now with all the gadgets and lemmas introduced above, we are ready to prove Proposition~\ref{prop:CFGLInflow}.
\begin{proof} [Proof of Proposition~\ref{prop:CFGLInflow}]
Given an arc $a=(u,v)\in A$ and an extreme point $f$, we first construct a spanning tree $\bT_f$ with Property~\ref{property:accessible} by Algorithm~\ref{alg:spanning}. Then the flow over arc $a$
\begin{subequations}
\begin{align}
f_a &= \sum_{v_{l,t} \in \bT_{f,1}} b_{l,t} \\
& =  \sum_{v_{1,t} \in \bT_{f,1}} b_{1,t} + \sum_{l=2}^L\sum_{v_{l,t} \in \bT_{f,1}} b_{l,t}, \label{eq:flowa1}
\end{align}
\end{subequations}
by~\eqref{eq:flow:a}. By Lemma~\ref{lemma:row1}, the first term in~\eqref{eq:flowa1}, $\sum_{v_{1,t} \in \bT_{f,1}} b_{1,t}$, equals $\sum_{i \le t \le j} b_{1,t}$ or $\sum_{1 \le t \le i} b_{1,t} + \sum_{j \le t \le T} b_{1,t}$, for some $i,j\in [T]$ with $i\le j$. Thus the term $\sum_{v_{1,t} \in \bT_{f,1}} b_{1,t}$ takes $O(T^2)$ values in all extreme points of $P_F$, and these values can be enumerated in $O(T^2)$ time. For the second term in~\eqref{eq:flowa1}, by Lemma~\ref{lemma:rowl}, $\sum_{v_{l,t} \in \bT_{f,1}}b_{l,t}$ equals either $\sum_{i_1 \le t \le j_1} b_{l,t} + \sum_{i_2 \le t \le j_2} b_{l,t}$ or $\sum_{t\in [T]} b_{l,t}-(\sum_{i_1 \le t \le j_1} b_{l,t} + \sum_{i_2 \le t \le j_2} b_{l,t})$ for some $i_1, j_1, i_2, j_2 \in [T]$ with $i_1\le j_1 \le i_2 \le j_2$. Thus the term $\sum_{v_{l,t} \in \bT_{f,1}} b_{l,t}$ takes $O(T^4)$ values, and these values can be enumerated in $O(T^4)$ time.

Therefore, the flow $f_a$ takes $O(T^2)\times O(T^{4(L-1)}) = O(T^{4L-2})$ values in all extreme points of $P_F$, and these values can be enumerated in $O(T^{4L-2})$ time.
\end{proof}

\subsection{NP-hard cases}
Each condition in Corollary~\ref{cor:mfgl} is crucial for the U-MFG to be polynomially solvable. As shown below, the U-MFG becomes NP-hard if any of those conditions is relaxed.

\begin{proposition}\label{thm:CFGhard} The U-MFG with sources in one row and a varying number of rows is NP-hard.
\end{proposition}
\begin{proof} Our proof is based on a reduction from the partition problem~\cite{garey1979computers}. Given a partition instance with a set of $n$ integers $\{y_1,\ldots,y_n\}$, we construct a U-MFG instance with $n+1$ rows, $n+1$ columns, two sources in row one, and $n$ sinks, as shown in Figure~\ref{fig:CFGLNPhard}. The grid has two sources ($v_{1,1}$ and $v_{1,2}$, each with supply $\sum_{i \in [n]}{y_i}/2$) and $n$ sinks ($v_{i+1,i+1}$ with demand $y_i$ for each $i\in [n]$). The cost over each incoming arc of sink $v_{i+1,i+1}$ ($(v_{i+1,i},v_{i+1,i+1})$ or $(v_{i,i+1}, v_{i+1,i+1})$) is 1 for sending nonzero flow and 0 otherwise for each $i\in [n]$. Each of the remaining arcs in Figure~\ref{fig:CFGLNPhard} has zero cost of sending any flow. 
\begin{figure}[htb]
    \centering
    \includegraphics[scale=0.6]{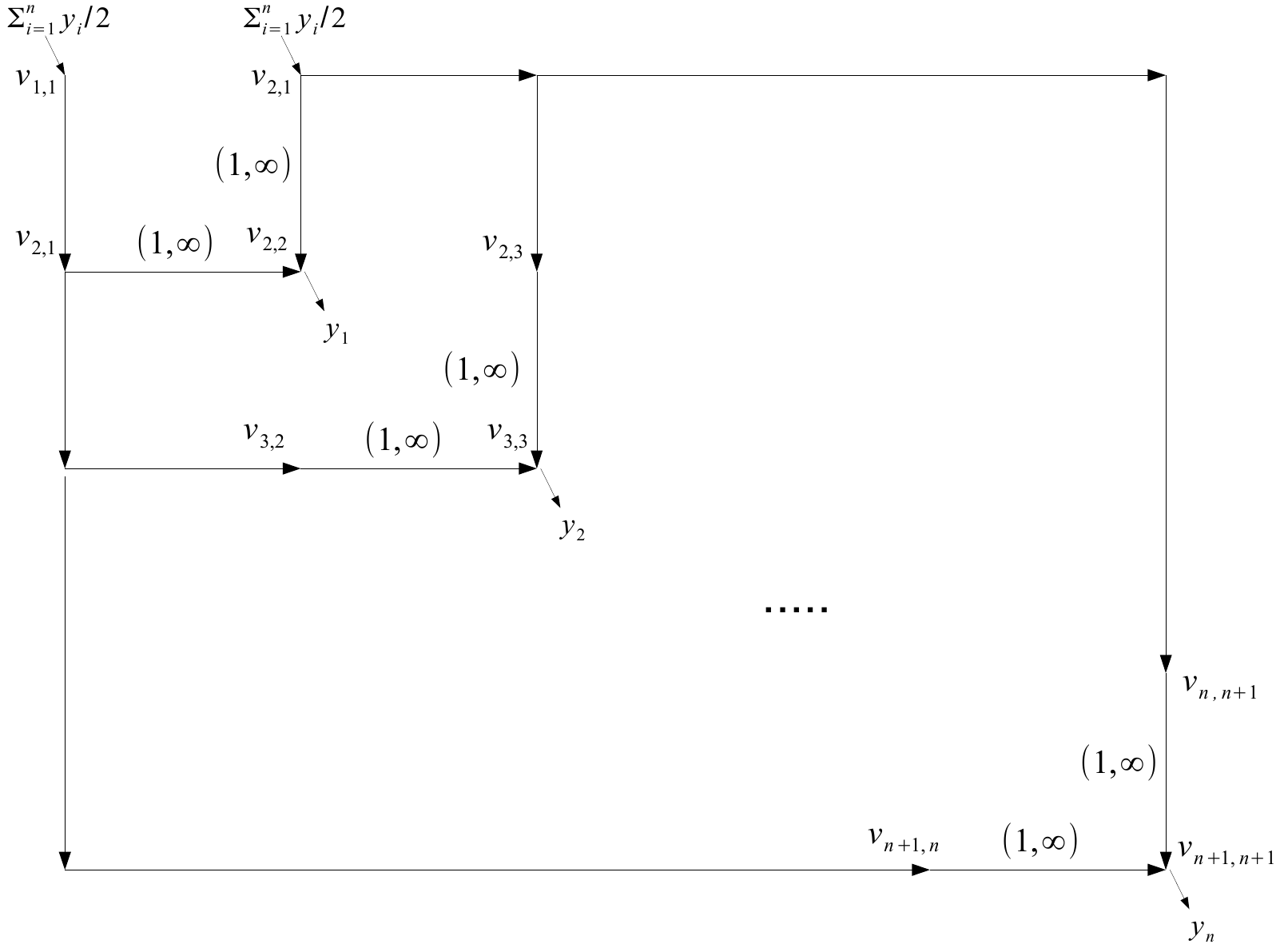}        
    \caption{The MFG instance with varying $L$}
    \label{fig:CFGLNPhard}
\end{figure}

We claim that the partition instance is a yes instance if and only if the optimal objective value of the constructed U-MFG instance is $n$. First notice the optimal objective value of the U-MFG instance is at least $n$, since at least one incoming arc of sink $v_{i+1,i+1}$ needs to carry nonzero flow, for each $i\in [n]$. If there exists a subset $S\subseteq [n]$ such that $\sum_{i\in S} y_i =\sum_{i \in [n]}{y_i}/2$, then we can create a feasible flow in the U-MFG instance as follows: send flow $y_i$ over the forward arc $(v_{i+1,i}, v_{i+1,i+1})$ to satisfy the demand $y_i$ for $i\in S$, and send flow over the downward arc $(v_{i,i+1}, v_{i+1,i+1})$ for each $i \notin S$. The flow value in each remaining arc can be computed by the flow conservation constraints. This feasible flow has a cost of exactly $n$. On the other hand, suppose $f^*$ is an optimal flow of the U-MFG instance with cost $n$. Then exactly one incoming arc of sink $v_{i+1,i+1}$ has nonzero flow, for each $i\in [n]$. Let $S$ be the set of indices $i\in [n]$ such that the flow over arc $(v_{i+1,i}, v_{i+1,i+1})$ in $f^*$ is strictly greater than 0. Then by the flow conservation constraints, $\sum_{i\in S}y_i = \sum_{i\in [n]} y_i/2$.
\end{proof}

\begin{proposition} \label{thm:cfg3NPhard}
The U-MFG with a constant number of rows and both sources and sinks in at least two rows is NP-hard.
\end{proposition}
\begin{proof}
Our proof is also based on a reduction from the partition problem. Given a partition instance with a set of $n$ integers $\{y_1,\ldots,y_n\}$, we construct a U-MFG instance with three rows and $2n$ columns, as shown in Figure~\ref{fig:CFG3NPhard}. Let $D$ be an integer larger than $\sum_{i \in [n]}{y_i}$, say, $2\sum_{i \in [n]}{y_i}$. The grid has $n+1$ sources ($v_{1,1}$ with supply $\sum_{i\in [n]} y_i/2$ and $v_{2,2i-1}$ with supply $D$ for each $i\in [n]$) and $n+1$ sinks ($v_{2,2i}$ with demand $y_i$ for each $i\in [n]$ and $v_{3,2n}$ with $nD-\sum_{i\in [n]} y_i/2$). The cost over each incoming arc of sink $v_{2,2i}$ ($(v_{2,2i-1},v_{2,2i})$ or $(v_{1,2i}, v_{2,2i})$) is 1 for sending nonzero flow and 0 otherwise, for each $i\in [n]$. Each of the remaining arcs in Figure~\ref{fig:CFG3NPhard} has zero cost of sending any flow. 
\begin{figure}[htb]
    \centering
    \includegraphics[scale=0.75]{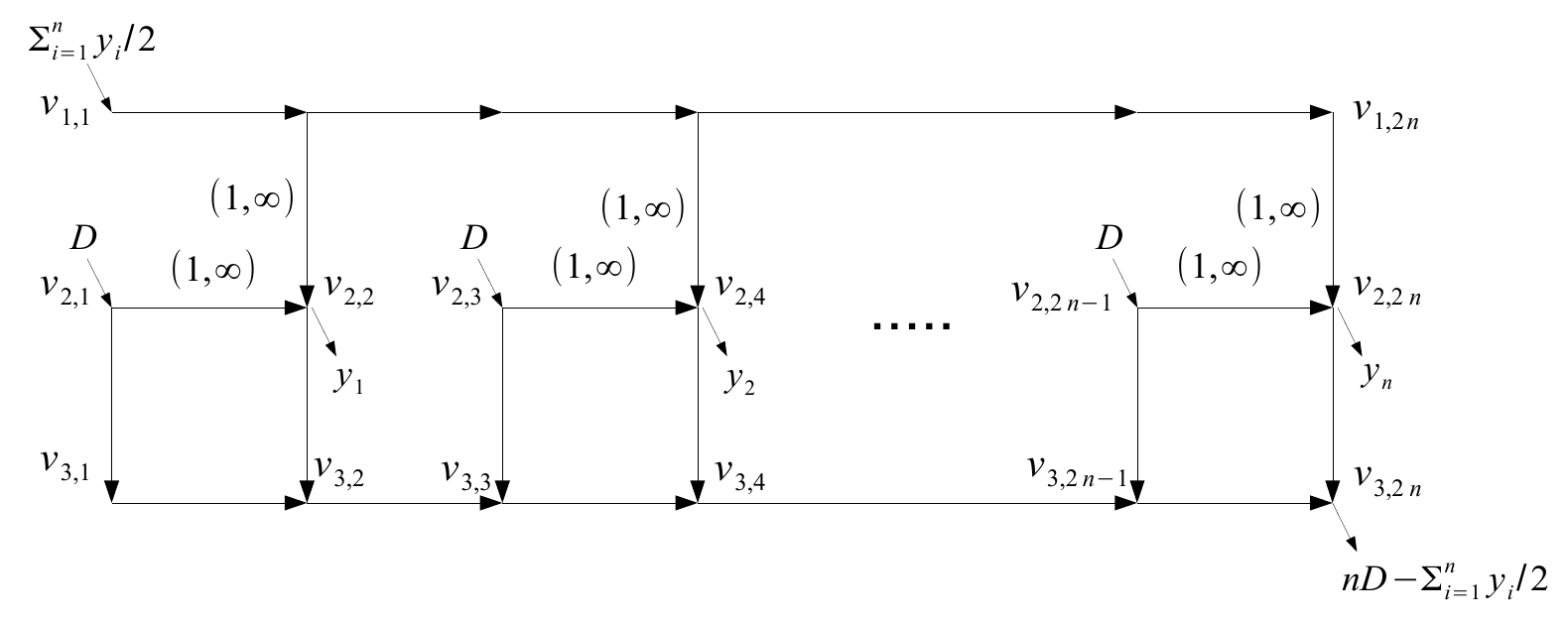}        
    \caption{The MFG instance with both sources and sinks in two rows}
    \label{fig:CFG3NPhard}
\end{figure}
It can be seen that the partition instance is a yes instance if and only if the optimal objective value of the constructed U-MFG instance is $n$, following a similar argument as in the proof of Proposition~\ref{thm:CFGhard}.
\end{proof}

\section{Extensions to grids with additional arcs} \label{sec:extension}
First notice the grid with only backward arcs (instead of forward arcs) or upward arcs (instead of downward arcs) is isomorphic to the grid in Figure~\ref{fig:grid}, so the minimum-concave-cost flow problem over this type of grid can be solved in the same way as the grid with only forward and downward arcs. We now consider the grid with both backward and forward arcs and/or both upward and downward arcs. Let us call the MFG with both back and upward arcs the MFG-BU.

The reformulation and algorithm in Section~\ref{sec:DP} can be easily modified for the MFG-BU, but the computational complexity results in Section~\ref{sec:CFG} and Section~\ref{sec:UCFG} vary across different cases. We illustrate the differences below.

First the dynamical system reformulated from the MFG-BU will include flow over backward arcs into states and flow over upward arcs into decision variables. In particular, the state $\BFs^t$ at stage $t$ becomes a $2L$-dimensional vector whose first $L$ components denote the flows over forward arcs $(v_{l,t},v_{l,t+1})$ and last $L$ components denote the flows over backward arcs $(v_{l,t+1},v_{l,t})$, for $l\in [L]$; the decision variable $\BFu^t$ at stage $t\in [T]$ becomes a $2(L-1)$-dimensional vector whose first $(L-1)$ components denote the flows over downward arcs $(v_{l,t+1},v_{l+1,t+1})$ and last $(L-1)$ components denote the flows over upward arcs $(v_{l+1,t+1},v_{l,t+1})$, for $l\in [L-1]$.	The system equations and cost calculation need to be modified accordingly. The algorithm for the MFG-BU remains almost the same as Algorithm~\ref{alg:mfg}, with the only change being to include costs over backward and upward arcs in the calculation of arc costs; its complexity is slightly higher than that in Theorem~\ref{prop:runtime} (up to a constant factor in the exponent), and the analysis is similar.

Theorem~\ref{thm:ccfgl_p} no longer holds for the capacitated MFG-BU. The key reason is that we cannot assume sources and sinks are in row one and row $L$ (or the boundary of the outer face of the grid) any more, due to the existence of upward arcs. Thus Lemma~\ref{lem:intersect} cannot be applied, and the sources and sinks in subtree $\bT_{f,1}$ do not appear consecutively. Corollary~\ref{cor:ccfg} still holds though for the capacitated MFG-BU if there is no upward arc in the grid. Similarly, the uncapacitated MFG-BU with no upward arcs, sources in one row and a constant number of rows is still polynomially solvable. It becomes NP-hard when the grid contains upward arcs, as shown below.

\begin{proposition}
The U-MFG with upward arcs and three rows is NP-hard.
\end{proposition}
\begin{proof}
Our proof is based on a reduction from the partition problem. Given a partition instance with a set of $n$ integers $\{y_1,\ldots,y_n\}$, we construct a U-MFG instance with three rows and $n+1$ column, as shown in Figure~\ref{fig:CFGUNPhard}. The grid has two sources ($v_{1,1}$ with supply $\sum_{i\in [n]}y_i/2$ and $v_{1,2}$ with supply $\sum_{i\in [n]} y_i/2$) and $n$ sinks ($v_{2,i+1}$ with demand $y_i$ for each $i\in [n]$). The cost over each incoming arc of sink $v_{2,i+1}$ ($(v_{1,i+1},v_{2,i+1})$ or $(v_{3,i+1}, v_{2,i+1})$) is 1 for sending nonzero flow and 0 otherwise, for each $i\in [n]$. Each of the remaining arcs in Figure~\ref{fig:CFGUNPhard} has zero cost of sending any flow. 
\begin{figure}[htb]
    \centering
    \includegraphics[scale=0.7]{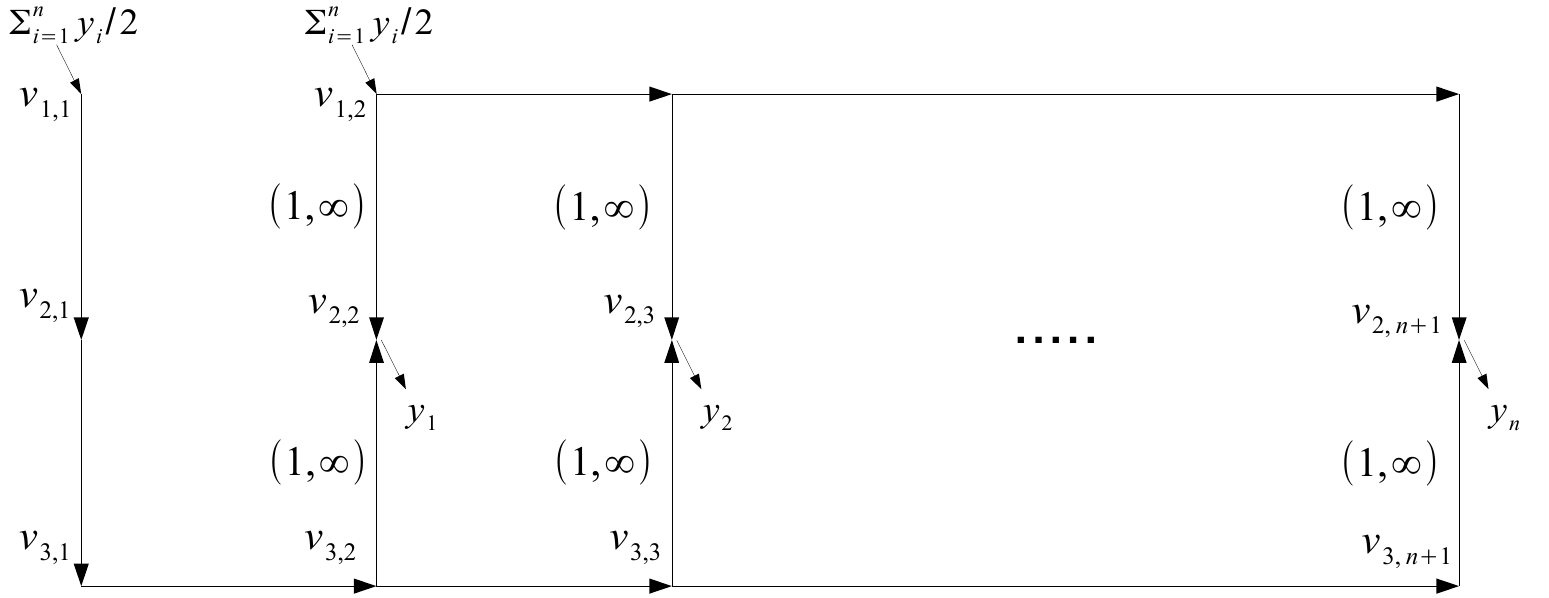}        
    \caption{The MFG instance with upward arcs}
    \label{fig:CFGUNPhard}
\end{figure}
It can be seen that the partition instance is a yes instance if and only if the optimal objective value of the constructed U-MFG instance is $n$, following a similar argument as in the proof of Proposition~\ref{thm:CFGhard}.
\end{proof}

We can even extend our polynomially solvable results to grids with diagonal arcs of the form $(v_{l,t}, v_{l+1,t+1})$ or $(v_{l,t+1}, v_{l+1,t})$, as long as adding these arcs keeps the grid planar. Then Lemma~\ref{lem:intersect} still holds, and the reformulation and analysis of the algorithm can be done in a similar fashion as the inclusion of backward arcs. Finally, throughout the paper we studied the computational complexity of the MFG when the number of \emph{columns} in the grid, $T$, varies. By symmetry, if we let the number of rows $L$ varies and put the corresponding conditions on columns (such as all sources are in one column), all results hold accordingly.

\section{Conclusions} \label{sec:conclude}
We establish a full characterization of the computational complexity of the minimum-concave-cost flow problem in a two-dimensional grid, based on the number of rows of the grid, the number of different capacities over all arcs, and the location of sources and sinks. We develop polynomial-time algorithms for the problem under several general conditions, by exploiting the fact of the grid being planar and the combinatorial structure underlying the extreme points of the flow polyhedron. We also complement these results with hardness results when any of the conditions is relaxed. Our results answer several open questions raised in the lot sizing and supply chain literature. An interesting open problem left is the computational complexity of the C-MFG with constant $K$, varying $L$, and sources and sinks in at most two rows, which we conjecture to be NP-hard.
\bibliographystyle{abbrv}
\bibliography{CFG}   
\end{document}